\documentclass[runningheads]{llncs}
\usepackage{graphicx}
\usepackage{mathtools}
\usepackage{cite}
 \usepackage{algorithm}
 \usepackage{amssymb}
 \usepackage{subcaption}
 \usepackage{algorithmic}
 \usepackage{amsmath}
 \usepackage{mathrsfs}
 \usepackage{yhmath}
 \usepackage{epsfig}
\usepackage{color,soul}

\makeatletter
\renewcommand\section{\@startsection{section}{1}{\z@}%
            {-8\p@ \@plus -4\p@ \@minus -4\p@}%
            {6\p@ \@plus 4\p@ \@minus 4\p@}%
            {\normalfont\large\bfseries\boldmath
            \rightskip=\z@ \@plus 8em\pretolerance=10000 }}
\renewcommand\subsection{\@startsection{subsection}{2}{\z@}%
            {-8\p@ \@plus -4\p@ \@minus -4\p@}%
            {6\p@ \@plus 4\p@ \@minus 4\p@}%
            {\normalfont\normalsize\bfseries\boldmath
            \rightskip=\z@ \@plus 8em\pretolerance=10000 }}
\renewcommand\subsubsection{\@startsection{subsubsection}{3}{\z@}%
            {-4\p@ \@plus -4\p@ \@minus -4\p@}%
            {-1.5em \@plus -0.22em \@minus -0.1em}%
            {\normalfont\normalsize\bfseries\boldmath}}
\makeatother

\usepackage{etoolbox}

\makeatletter
\expandafter\patchcmd\csname\string\algorithmic\endcsname{\itemsep\z@}{\itemsep=0ex plus4pt}{}{}
\makeatother

\usepackage{enumitem}
\setlist[description]{noitemsep, topsep=0pt}

\definecolor{mygray}{gray}{0.1}

\makeatletter
\expandafter\patchcmd\csname\string\algorithmic\endcsname{\itemsep\z@}{\itemsep=-0.1ex plus4pt}{}{}
\makeatother

 \newtheorem{observation}{Observation}
\renewenvironment{proof}      
{
 \noindent{\bf Proof.} 
}
{
 \hfill$\qed$
}
\newenvironment{proofsketch}      
{
 \noindent{\bf Proof Sketch.} 
}
{
 \hfill$\qed$
}

\newcommand{\spp}{\mathcal{U}}
\newcommand{\prs}{\mathcal{W}^a}
\newcommand{\qp}{\sigma}
\newcommand{\qs}{\overline{\sigma}}
\newcommand{\range}{r}
\newcommand{\wv}{\omega}
\newcommand{\norm}[1]{\|#1\|_{_1}}

\newcommand{\vr}{\mathcal{V}_{max}}

\begin{document}
\title{Rectilinear Shortest Paths Among Transient Obstacles}
%
%
\author{Anil Maheshwari \and
Arash Nouri \and
J\"org-R\"udiger Sack}
\authorrunning{Anil Maheshwari \and
Arash Nouri \and
J\"org-R\"udiger Sack}
%
\institute{School of Computer Science, Carleton University, Ottawa, Canada \\
\email{\{anil,arash,sack\}@scs.carleton.ca}}
\maketitle       
\begin{abstract}

This paper presents an optimal $\Theta(n \log n)$ algorithm for determining  time-minimal rectiliear paths among $n$ transient rectilinear obstacles. An obstacle is transient if it exists in the scene only for a specific time interval, i.e., it appears and then disappears at specific times.
Given a point robot moving with bounded speed among transient rectilinear obstacles and a pair of points $s$, $d$, we determine a time-minimal, obstacle-avoiding path from $s$ to $d$. The main challenge in solving this problem arises as the robot may be required to wait for an obstacle to disappear, before it can continue moving toward the destination. 
Our algorithm builds on the continuous Dijkstra paradigm, which simulates propagating a wavefront from the source point. 
We also solve a query version of this problem. For this, we build a planar subdivision with respect to a fixed source point, so that  minimum arrival time to any query point can be reported in $O(\log n)$ time, using  point location for the query point in this subdivision.

\keywords{Shortest Path \and Transient Obstacles \and Time Minimal Path \and Time Discretization \and Continuous Dijkstra.}
\end{abstract}
\section{Introduction}

We study a variant of the classical shortest path problem in which each obstacle exists only during a specific time interval. Such obstacles are called \textit{transient obstacles} (see e.g., \cite{fuji0}). Besides solving an interesting problem in itself, our solutions may find applications in other motion planning problems in time-dependent environments. Transient obstacles can e.g., be used to approximate dynamic obstacles in the plane \cite{fuji-pixel, LaValle}.
In such settings, the trajectories of the moving obstacles are divided into a set of small pieces. Each piece is treated as a transient obstacle that exists in the scene only for the time interval in which the moving obstacle and the piece intersect. The approximation quality can be adjusted by varying the sizes of the pieces.
This adequately models real world scenarios in which  robots are limited by the sampling rate of their sensors  acquiring information and executing motion commands.

In general, our model considering transient obstacles can be useful for applications where one can define a discretized representation of time by a set of stages.
 For instance, in the area of \textit{path planning under uncertainty}, one considers the following problem:
 Let $\{R_1,..., R_n\}$ be a set of regions, where each region becomes contaminated at a random time. The probability at which $R_i$ is contaminated at time $t$, is given by a probability distribution $P_i(t)$. In such a setting, a natural approach is to search for a shortest path which is contamination-free with high probability. This is a class of motion planning referred to as \textit{hazardous region and shelter problems} \cite{lavalle_2006}. 
 A suitable means of planning a low contamination path, is to bound the probability at which the intersecting regions are contaminated. More precisely, for a small value of $\epsilon \in [0, 1]$, the robot cannot enter a region $R_i$ if $P_i(t) > \epsilon$. This can be viewed as a time discretization into a set of ``high risk'' time intervals for the regions. Using the corresponding probability distribution,
 we can determine a time interval $T_i$ (or in some cases more than one), which contains the contamination time with a probability of $1 - \epsilon$. This problem is easily transformed into our model where the confidence intervals are mapped into existence intervals for the transient obstacles.

\textbf{Related work.} The shortest path problem among transient obstacle was first studied by Fujimura \cite{fuji0}, who presented an $O(n^3 \log n)$ time algorithm for finding a  time-minimal path among transient (non-intersecting) polygonal obstacles. Later \cite{fuji1}, he proposed an $O(n^4)$ time algorithm for a variant of this problem in which the obstacles are allowed to occupy the same area of the plane (i.e., intersecting obstacles). A recently introduced model \cite{violation} considers another variation of this problem, where the path is allowed to pass through $k$ obstacles. They present an $O(k^2 n \log n)$ time algorithm, where $n$ is the total number of obstacle vertices. A more complex version of this problem has been studied in \cite{removeable}, in which the robot may pass through obstacles at some cost. They proved that this problem is NP-hard even if the obstacles are vertical line segments. 

\textbf{Our Contributions.} In this paper, we present an optimal $\Theta(n \log n)$ time algorithm for computing a time-minimal rectilinear path among rectilinear transient obstacles. Although our problem is a special case of the shortest path problem among transient obstacles, the methodology and the results of this work also have the potential to lead to an improvement of the existing $O(n^3 \log n)$ time algorithm for the general case. 
We first discuss a simple problem instance in which  the given obstacles are rectilinear segments. Then, we generalize the  algorithm developed for the simpler setting to simple rectilinear  polygons. Section \ref{prel} describes preliminaries, definitions and introduces some notation. Section \ref{range-search} presents several techniques that are subsequently employed in this  paper. 
Building on these techniques, Section \ref{alg} presents an $O(n^2 \log n)$ time algorithm for the problem, which is already an improvement over the existing algorithm applied to  our setting. Finally, Section \ref{eff-sec} details our optimal $\Theta(n \log n)$ time algorithm.  

\section{Preliminaries} \label{prel}

Let $E=\{E_1,...,E_n\}$ be a set of rectilinear transient edges, where each edge $E_i \in E$ exists in the scene during a time interval $[T_i^a, T_i^d]$, where $0 \leq T_i^a < T_i^d$. The edges are disjoint, i.e., no two edges are allowed to overlap at any time. We assume that the edges are in general position, which means that, no two edges lie on a common line.
Let $\mathcal{R}$ be a point robot having maximum speed $\vr$.
 For two given points $s$ and $d$ in the plane, our problem is to determine a time-minimal rectilinear path for the point robot from $s$ to $d$, denoted by  $\pi(s, d)$, which is collision free, i.e., the point robot does not pass through the edges during their existence intervals. W.l.o.g., we assume the robot always departs from $s$ at time $0$.

Our strategy is to employ the ``continuous Dijkstra'' paradigm \cite{suri-cd,mitchell-cd}, which has been applied to solve numerous shortest path problems among permanent (i.e., non-transient) obstacles \cite{suri-cd,mitchell-cd,Mi3,Mi2}. We provide here a brief description of this paradigm.
The continuous Dijkstra's technique models the effects of sweeping an advancing wavefront from the source point till it reaches the destination. A \textit{wavefront} (in $L_1$ metric space) is defined as the set of points on the plane at equal $L_1$ distance from the source. Initially, the wavefront is point located at $s$. After a short time period, it becomes a rhombus centered at $s$ with diameter $\epsilon$, where $\epsilon$ is a small positive constant (see Figure \ref{wavefront-fig} (a) for an illustration). The continuous Dijkstra's algorithm proceeds by expanding the rhombus outward from its center point. 
At any point in time, the wavefront consists of a set of line segments, known as \textit{wavelets}. A wavelet is defined as a maximal set of points on the wavefront, such that each point on the wavelet has a shortest path from $s$ via a common vertex. Each wavelet originates at a vertex, which is called the \textit{source} of the wavelet. Therefore, each wavelets moves in one of the four fixed directions: \textit{north-east, north-west, south-east, south-west}. We abbreviate these four directions as \{NE, NW, SE, SW\}. 
More precisely, the wavelets are in four fixed inclinations with respect to the $x$-axis with angles: $\pi / 4$, $3 \pi / 4$, $5\pi / 4$ and $7\pi/4$.

For our setting, we need to modify   the continuous Dijkstra's model described above for metric shortest paths, to time-minimal paths among transient obstacles. Note that, on a time-minimal path, the robot's speed alternates between $\vr$ (i.e., the robot is moving) and zero (i.e, the robot is waiting). It is easily seen that, by arriving earlier at some obstacles and then waiting there until the obstacle disappears, the robot can avoid any speed other than $\vr$ and zero. After waiting  the robot  continues to move towards the next destination. The points on the boundary of  obstacles where robots may be waiting are called \textit{wait points}. 
The behavior of the wavefront changes at  portions of  obstacles that are ``potential'' wait point candidates.
 
Given a point $p$ in the plane, we say $\pi(s, p)$ intersects an edge $E_i$, if it has a wait point on $E_i$; and we say it intersects a vertex $v$ if $v \in \pi(s, p)$.
Let $S(\pi(s, p))$ be the sequence of  edges and  vertices that $\pi(s, p)$ intersects. We formally define a wavelet as follows.
 
\begin{definition}\label{ws}
A \textbf{wavelet} $\wv$ is a maximal set of points, such that for each pair of points $p,q \in \wv$ there exist two paths $\pi(s, p)$ and $\pi(s, q)$ with equal arrival times, for which $S(\pi(s, p)) = S(\pi(s, q))$. Let $x$ be the last element in $S(\pi(s, p))$. We say $x$ is the \textbf{origin} of $\wv$ (or alternatively, we say that $\wv$ is \textbf{originating} from $x$). If $x$ is an edge,
 we say $\wv$ is a \textbf{segment wavelet}; otherwise, $\wv$ is a \textbf{point wavelet}. A \textbf{wavefront} is defined as the union of the wavelets at an equal time.
\end{definition}

By the above definition, similar to the original version of continuous Dijkstra, the point wavelets are in four fixed inclinations with respect to the $x$-axis with angles: $\{\pi / 4, 3 \pi / 4, 5\pi / 4, 7\pi/4\}$. Now, observe the following property of the time-minimal paths in our setting.

\begin{observation}\label{prepan-obs} {\normalfont \cite{fuji1}} 
When, after waiting on an edge $E_i$, the robot departs, at some time $T_i^{d}$, it will use  a move perpendicular to the orientation of $E_i$  (see Figure \ref{wavefront-fig} for an illustration).
\end{observation}

\begin{figure}[H]
\captionsetup[subfigure]{justification=centering}
\centering
 \begin{subfigure}{.4\linewidth}
 \centering
\includegraphics[scale=0.8]{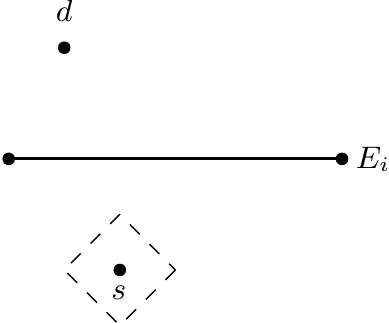}
  \caption{}
 \end{subfigure}
 \begin{subfigure}{.45\linewidth}\centering
\includegraphics[scale=0.8]{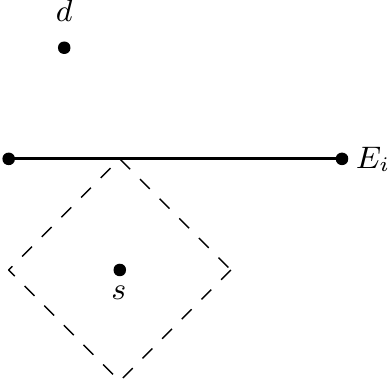}
  \caption{}
 \end{subfigure}\par\bigskip\par
 \begin{subfigure}{.45\linewidth}
 \centering
\includegraphics[scale=0.8]{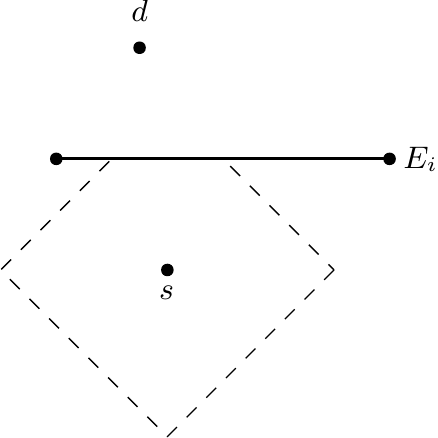}
  \caption{}
 \end{subfigure}
 \begin{subfigure}{.45\linewidth}\centering
\includegraphics[scale=0.8]{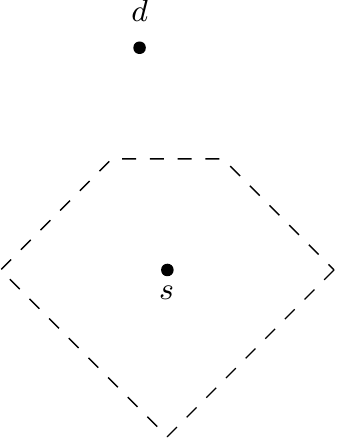}
  \caption{}
 \end{subfigure}\par\bigskip\par
 \begin{subfigure}{.45\linewidth}\centering
\includegraphics[scale=0.5]{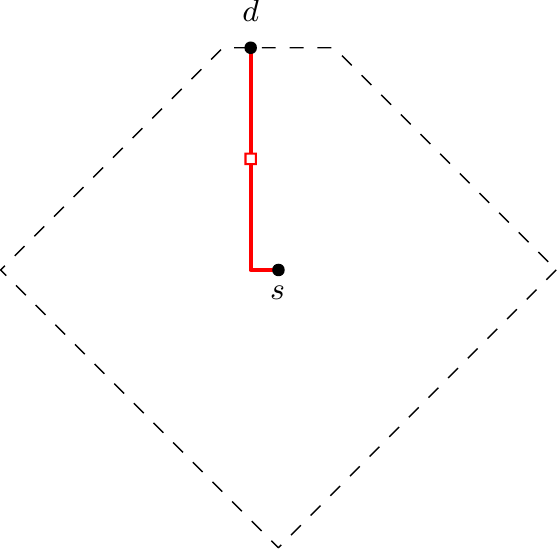}
  \caption{}
 \end{subfigure}
\caption{Initially, the wavefront is point located at $s$. (a) After a short time period, it becomes a rhombus centered at $s$. (b) The wavefront hits the body of $E_i$ during its existence time interval. (c) The north accessible segment of  the source point $(s, 0)$ is located on $E_i$. (d) At the
depicted instance, $E_i$ disappears. (e) The wavefront continues expanding until it hits $d$. A time-minimal path from $s$ to $d$ is represented by red. The red square represents the wait point of this path.} \label{wavefront-fig}
\end{figure}

By the above observation, each segment wavelet  propagates outwards perpendicularly to its orginating segment. Since the edges are axis-parallel, each segment wavelet is oriented in one of four (axis-parallel) directions: $\{0, \pi / 2, \pi, 3\pi/4\}$ and moves in one of four directions: \textit{north, south, east or west}; these are abbreviated as:
\{N, S, E, W\}, respectively. 


Our algorithm  propagates  a wavelet $\wv$, with  inclination $\theta$, outwards by using  a sweep line through $\wv$     (refer to Section \ref{tcp} for details). When $\wv$ encounters (or ``hits'') an obstacle, we add new wavelets originating from a vertex or an edge of the obstacle. For each wavelet $\wv$, we designate an area called \textit{search region}, from which $\wv$ propagates its interior (refer to Section \ref{seg} for a formal definition). The time at which  propagation starts is called  \textit{departure time}. 
A key property that we will be subsequently using is that wavelets are line segments with fixed inclinations. 
This enables us to efficiently find the next propagation``events''   using range searching queries (see Section \ref{range-search}) (events are intersections between the wavelets and the obstacles). 

We say a path is \textit{monotone} if any axis-parallel line intersects the path in at most one connected set. For any pair of consecutive vertices $u$ and $v$ on a shortest path among non-transient obstacles, in \cite{origi-l1}, it is proven that the sub-path from $u$ to $v$ is monotone. In the following lemma, we show that the analogous property also holds for time-minimal rectilinear paths among transient obstacles.

\begin{lemma}\label{monotone} {\normalfont \bf (Monotonicity Property).}
In our model, let $u$ and $v$ be two consecutive vertices on $\pi(s, d)$.
The sub-path of $\pi(s, d)$ from $u$ to $v$, denoted by $\pi(u, v)$, is monotone.
\end{lemma}
\begin{proof}
If $\pi(u, v)$ contains no wait points, then,
with the same proof as in \cite{origi-l1}, $\pi(u, v)$ 
is a monotone path from $u$ to $v$. Furthermore, if $\pi(u, v)$ contains a set of wait points $W=\{p_1, ..., p_h\}$ sorted by their departure times, with the same proof as in \cite{origi-l1}, $\pi(u, p_1)$ is monotone. By Observation \ref{prepan-obs}, every wait point in $W$ is located on a same vertical or horizontal line. So, $\pi(p_1, v)$ is monotone and consequently, $\pi(u, v)$ is monotone.
\end{proof}

Define a pair $(p, t)$,  as \textit{point source}, denoted by $\qp(p, t)$, where $p{=}(X_p, Y_p)$ is the $x$-$y$ location of the robot at time $t$. We will simply say a path from $\qp$ instead of a path leaving $p$ at time $t$.
Suppose the robot departs  from $\qp$ by  moving north at maximum speed. For that motion, we define the \textit{north stop point} for $\qp$, denoted by $\spp(\qp, N)$, as the first point on any obstacle that the robot ``hits'' during the respective obstacles' existence times. Analogously, we define $\spp(\qp, S)$, $\spp(\qp, E)$ and $\spp(\qp, W)$ as the south, east and west stop points, respectively. Note that the locations of the stop points may change depending on the departure time $t$. 

 A point $q \in E_i$ is called \textit{accessible} from point source $\qp(p, t)$, if there exists a time-minimal path from $\qp$ to $q$, denoted by $\pi(\qp, q)$, such that: (1) $\pi(\qp, q)$ contains no wait points and (2) the robot arrives at $q$ during the existence time interval of $E_i$. We denote by $T(\qp, q) = t + {\norm{pq} \over \vr}$ the arrival time of this path, where $\norm{pq}$ is the $L_1$ distance between the two points. Observe that, any stop point for $\qp$ is an accessible point from $\qp$. 
 Let $\spp(\qp, N) \in E_i$ be the north stop point for $\qp$. 
 We define the \textit{north accessible segment} of $\qp$ as a maximal set of accessible points on $E_i$. Note that, $\spp(\qp, N)$ is a point on the north accessible segment. Analogously, we define the other accessible segments of $\qp$ in the three other directions.

Given a sub-segment $e$=$((X_1, Y_1),(X_2, Y_2))$ of an edge $E_i$, we define a \textit{segment source} $\qs(e, T_i^d)$ 
${=}\cup_{p\in e}\{\qp(p, T_i^d)\}$ as a maximal set of point sources located on $e$ having common departure time  $T_i^d)\}$.

A point $q \in E_i$ is a \textit{north accessible point} for $\qs$, if there is exist $\qp(p, T_i^d) \in \qs(e, T_i^d)$ such that $q{=}\spp(\qp, N)$. 
We define the \textit{north stop segment} for $\qs$, denoted by $\spp(\qs, N)$, as the maximal set of north accessible points which have minimum distance to $e$. By the general position assumption, $\spp(\qs, N)$ is a connected sub-segment of an edge in $E$. Intuitively, if we drag the segment $e$ north, the north stop segment is the first intersection between the dragging segment and the obstacles.
Analogously, we define $\spp(\qs, S)$, $\spp(\qs, E)$ and $\spp(\qs, W)$.

\section{Range Searching Techniques} \label{range-search}
When propagating a wavelet, we wish to quickly determine the next event, where the wavelet intersects an obstacle. In this section, we present our techniques employed to solve this problem. First, we present a solution to the problem of determining the stop points for a query point source (see Lemma \ref{spqp}). Then, we  devise an algorithm to report the stop segments for a query segment source (see Lemma \ref{stop-point-query-segment}). Using these stop points, we define a rectangular range  (the search region), which  contains potential next  points/edges hit  by the wavelet. We identify these  using a range searching technique, presented in Lemma \ref{queryR}.

\subsection{Finding the Stop Points } \label{stop}
 
Given a query point source $\qp(p{=}(x,y), t)$, we denote by $\{\spp(\qp, N)$ $, \spp(\qp, S)$ $, \spp(\qp, E),$ $\spp(\qp, W)\}$ the four stop points for $\qp$. In this section, we present a solution to find $\spp(\qp, N)$ and other stop points can be found analogously. 
Let $\spp(\qp, N) \in E_i$, where $E_i=((X_1,Y),(X_2,Y))$. By definition, the robot hits $E_i$ during its existence time interval. So, we must have $T_i^a \leq t + {(Y-y) \over \vr} \leq T_i^d$. Hence, 
\begin{align} \label{y-ax}
    T_i^a - {Y \over \vr} \leq & t - {y \over \vr} \leq T_i^d - {Y \over \vr}.
\end{align}
Also, it is easily seen that, 
\begin{align} \label{x-ax}
    X_1 \leq x \leq X_2.
\end{align} 
As a result, if the north stop point for $\qp$ is located on $E_i$, Equations (\ref{y-ax}) and (\ref{x-ax}) must be satisfied. 
 These equations can be viewed as one equation in two dimensions, where the $y$ values are replaced by $(t - {y \over \vr})$: let $\range^s_i$ be a rectangle where $(X_1, T_i^a - {Y \over \vr})$ and $(X_2, T_i^d - {Y \over \vr})$ is  one of its opposite corner pairs. Given a point $\overline{p}=(x, t - {y \over \vr})$, observe that $\overline{p} \in \range^s_i$ if and only if $\qp$ satisfies the Equations (\ref{y-ax}) and (\ref{x-ax}). We call $\range^s_i$ the \textit{south shadow range} of $E_i$ and $\overline{p}$ the \textit{south shadow point} of $\qp$ (see Figure \ref{shadow} for an example).

\begin{figure}[t]
\captionsetup[subfigure]{justification=centering}
\centering
 \begin{subfigure}{.47\linewidth}
 \centering
\includegraphics[scale=0.4]{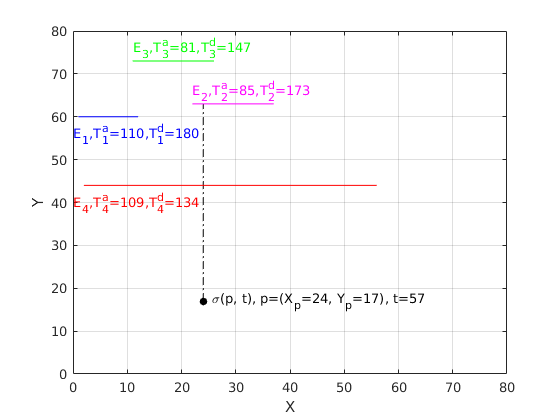}
  \caption{}
 \end{subfigure}
 \begin{subfigure}{.47\linewidth}\centering
\includegraphics[scale=0.4]{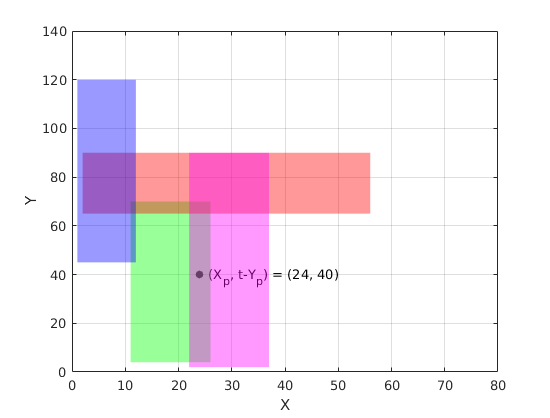} 
  \caption{}
 \end{subfigure}
\caption{(a) The north stop point for $\qp$ is located on $E_2$. For simplicity, we assume $\vr = 1$. Figure (b) illustrates the shadow ranges corresponding to the segment obstacles in Figure (a). Since the purple rectangle is the minimum weight shadow range that contains $\overline{p}$, the robot hits $E_2$ first.} \label{shadow}
\end{figure}

\begin{observation}\label{shadow-obs}
Let $\qp(
p, t)$ be a query point source and $\spp(\qp, N) \in E_i$ be its north stop point. Then, the south shadow point for $\qp$ is located inside the south shadow range of $E_i$, i.e., $\overline{p} \in \range^s_i$.
\end{observation}

Note that the reverse direction of the above observation does not always hold. In other words, there are several edges whose shadow ranges contain $\overline{p}$; however, only one includes the north stop point. Recall that a stop point represents the ``first'' intersection between the robot and the obstacles. So, $\spp(\qp, N)$ is located on an obstacle whose $Y$ value is minimum among all edges whose south shadow range contain $\overline{p}$ (see Figure \ref{shadow} (b)). Thus, we assign a weight to a shadow range $\range^s_i$, denoted by $\wv(\range^s_i)$, which is the $Y$ value of its corresponding (horizontal) edge $E_i$. In order to find the north stop point for the query point source $\qp$, we need to find the minimum weight shadow range $\range$ that contains $\overline{p}$.

\begin{lemma}\label{stb} 
 {\normalfont \cite{Agarwal-stab}} A set
$H$ of $n$ axis-parallel rectangles, where each rectangle $h \in H$ has a weight $\wv(h) \in \mathbb{R}$, can be maintained so that the minimum weight rectangle containing a query point can be
determined in $O(\log n)$ query time, after $O(n \log n)$ preprocessing time.
\end{lemma}

The following is the direct consequence of the above lemma.

\begin{lemma}\label{spqp} 
After $O(n \log n)$ time preprocessing, all stop points of a query point source can be found in $O(\log n)$ time.
\end{lemma}

\subsection{Finding the Stop Segments} \label{stopseg}

Let $\qs(
e,t)$ be a horizontal query segment source and $\spp (\qs, N)$ be its north stop segment, 
 where $e=((x_1, y),(x_2, y))$. Define $\overline{e}=((x_1, t - {y \over \vr}),(x_2, t - {y \over \vr}))$ as the \textit{shadow segment} of $e$. 
 By Equations (\ref{y-ax}) and (\ref{x-ax}) and Observation \ref{shadow-obs}, the following can be observed.

\begin{observation}\label{shadow-obs2}
Let $\qs(
e,t)$ be a horizontal segment source and $\spp (\qs, N) \in E_i$ be its north stop segment. Segment $\overline{e}$ intersects the south shadow range of $E_i$.
\end{observation}

 Recall that a stop segment is defined as the first intersection between the dragging segment $e$ and the obstacles. Thus, $\spp(\qs, N)$ is located on a segment $E_i$ whose south shadow range's interior intersects $\overline{e}$ and has minimum weight (i.e., $Y$ value). We now consider two cases: firstly, let $\overline{e}$ be located entirely inside the south shadow range of $E_i$. By definition, it is easily observed that the north stop point for any source point on $\qs(
e,t)$, is located on $E_i$. Thus, by locating a point on $\overline{e}$ we can find the stop segment for $\qs(
e,t)$ (see Lemma \ref{stb}). 

 For the second case, assume $\overline{e}$ intersects the boundary of $\range_i^s$. Let $R$ be the set of shadow ranges whose boundaries (horizontal and vertical line segments) intersect $\overline{e}$. 
By the following lemma, we can report the minimum weighted range $\range_i \in R$ whose boundary intersects the source segment $\overline{e}$.

\begin{lemma}\label{lemintersection} {\normalfont \cite{rectintersection}}
Given a family of $n$ rectilinear line segments $L$ and a query rectilinear line segment $s$,  $L$ can be preprocessed in $O(n \log n)$ time, so that a minimum weight segment in $L$ intersecting $s$, can be reported in $O(\log n)$ query time.
\end{lemma}

 Thus, the following lemma follows immediately.

\begin{lemma}\label{stop-point-query-segment} 
After $O(n \log n)$ time preprocessing, the stop segment for a query segment source can be found in $O(\log n)$ time.
\end{lemma}

\subsection{Range Searching for Minimum} \label{seg}

\begin{lemma} \label{queryR} {\normalfont \cite{Chazelle-MST}}
Let $H$ be a dynamic set of points in ${\rm I\!R}^2$ where insertions and deletions of the points are allowed. In $O(n \log n)$ time, we can preprocess $H$ into a data structure, so that, for a given query axis-parallel rectangle $r$, we can determine a minimum weight point inside $r \cap H$ in $O(\log n)$ time. $H$ can be updated in $O(\log n)$ time per insertion/deletion.
\end{lemma}

Let $V$ be the set of all vertices (end points of the edges in $E$) union $\{s,d\}$. Let $\range$ be a query rectangle and $p$ be one of its corners. We denote by $v_m$, a vertex of $V$ located inside $\range$ with minimum $L_1$ distance to $p$. W.l.o.g., assume $p$ is the bottom-left corner of $\range$. Let $B$ be an axis-parallel rectangular bounding box that contains all vertices in $V$. We assign a weight to each vertex $v\in V$, denoted by $\wv(v)$, which is the $L_1$ distance between $v$ and the bottom left corner of $B$. Observe that $v_m$ is the minimum weight vertex in $\range$. By Lemma \ref{queryR}, there is a data structure \cite{Chazelle-MST} that allows finding $v_m$ in $O(\log n)$ query time, after $O(n \log n)$ preprocessing time.

Let $\qp(
p, t)$ be a query point source and $\spp (\qp, N)$ and $\spp (\qp, E)$ be its north and east stop points, respectively. We define the \textbf{north-east search region} for $\qp$ as the rectangle where $\spp (\qp, N)$ and $\spp (\qp, E)$ are its two opposite corners (see Figure \ref{point-prop} (a) as an example). If $\spp (\qp, N)$ or $\spp (\qp, E)$ does not exist, we say that the north-east search region is \textit{undefined}. Thus, we can define, at most, four search regions corresponding to each query point source. For a horizontal (or vertical) segment source $\qs(
e, t)$, we define the \textbf{segment search region} of $\qs$ as a vertical (or horizontal) strip of width $|e|$ that entirely contains $e$. In the next section, we use the search regions to locate the events where the wavelets hit the obstacles.

\section{Algorithm} \label{alg}
 We design a simple data structure
 to represent the wavelets. 
 For each wavelet $\wv(q, t, \range)$, the data structure contains the following information:
\begin{itemize} 
  \item The source $q$, from which the wavelet is propagated. Recall that, wavelets with inclinations 
$\{0, \pi / 2, \pi, 3\pi/2\}$ originate from segment sources. Conversely, wavelets with inclinations 
$\{\pi / 4, 3\pi / 4, 5\pi / 4, 7\pi / 4\}$ originate from the point sources.
  \item The corresponding departure time of the wavelet, denoted by $t$. 
  \item A search region $\range$, which contains potential next intersections between the wavelet and the obstacles. In order to propagate $\wv$, we allow the wavelet to sweep the interior of $r$ and report ``hits'' by the wavelet. 
\end{itemize}

\subsection{Propagation} \label{tcp}
 
Propagating a wavelet $\wv(q, t, \range)$ means to allow the wavelet to sweep in its designated direction, until it hits a vertex $v$ (or alternatively, the body of an edge). We assume that the minimum arrival time at $q$ has been already calculated. Then, we calculate a potential minimum arrival time at $v$ using a shortest $L_1$ path from $q$ to $v$. This may involve deleting, updating, and creating wavelets corresponding to the advancing wavefront. Since the source of a wavelet is either a point or a segment, we present two algorithms to propagate these wavelets. 


Before discussing the propagation algorithms, we introduce some notation and give some definitions. Let $\qp(
p, t)$ be a point source and $\range^{NE}$ be its corresponding north-east search region. We denote by $\wv(\qp, t, \range^{NE})$ a wavelet that originates at $\qp$ and is propagating north-east inside $\range^{NE}$. Similarly, we can define (at most) three more wavelets, in the three directions $\{NW, SE, SE\}$, originating from $\qp$. Denote by $\prs(\qp)$ the set of wavelets originating from $\qp$ in all (at most) four directions. By Lemma \ref{spqp}, $\prs(\qp)$ can be found in $O(\log n)$ time, after $O(n \log n)$ preprocessing time. By Lemma \ref{spqp}, for each wavelet $\wv \in \prs(\qp)$, we can find the closest vertex to $\qp$ inside its corresponding region in $O(\log n)$ time, we denote this vertex by $\Gamma(\wv)$.

We design an algorithm called $PropagatePoint(\wv)$ (for details, see Algorithm \ref{propag}) which propagates a point wavelet $\wv(\qp,t,\range)$ inside its corresponding search region $\range$. There are four types of point wavelets depending on their directions (i.e, $NE$, $NW$, $SE$ and $SW$). W.l.o.g., we assume $\wv(\qp,t,\range)$ is propagating north-east; other directions can be treated analogously. 
Two types of events are discovered by this algorithm, which are explained next.

\begin{figure}[t]
\captionsetup[subfigure]{justification=centering}
\centering
 \begin{subfigure}{.45\linewidth}
 \centering
 \includegraphics[width=165pt]{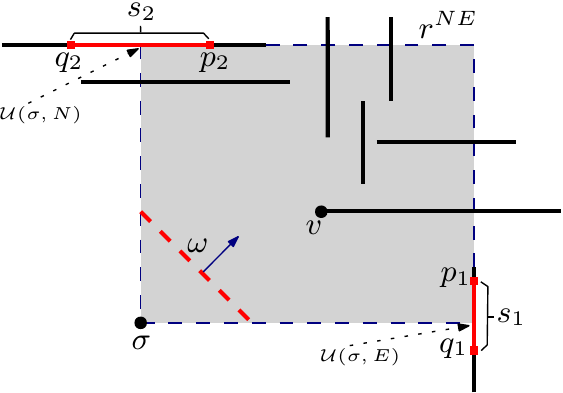}
  \caption{}
 \end{subfigure}
 \begin{subfigure}{.45\linewidth}\centering
    \vspace{0.5\baselineskip}
 \includegraphics[width=100pt]{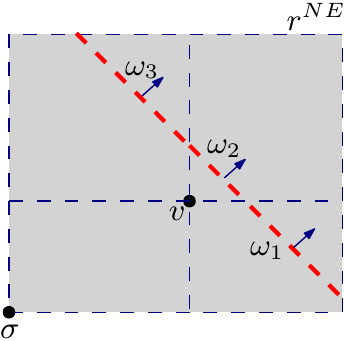}
    \vspace{0.9\baselineskip}
  \caption{}
 \end{subfigure}
\caption{The process of point propagation is illustrated. (a) $\wv$ is propagating north-east and $v$ is the first vertex it intersects; (b) when $\wv$ intersects $v$, it is split into three new wavelets $\wv_1$, $\wv_2$ and $\wv_3$. } \label{point-prop}
\end{figure}

Firstly, the algorithm finds all accessible segments of $\qp$, on the boundary of $\range$. For an example, see $s_1$ and $s_2$ in Figure \ref{point-prop} (a). Since the vertices are located in general position, there are at most four accessible segments on the boundary of $\range$ (one for each edge of $\range$). For each discovered accessible segment,
the algorithm adds two types of wavelets to the queue: (1) a segment wavelet whose departure time is the disappearance time of its associated edge and (2) a set of point wavelets originating from the end points of the segment (in Figure \ref{point-prop} (a), these points are denoted by $q_1, p_1, q_2$ and $p_2$). 
These two types of wavelets are added in lines \ref{seg-ln} and \ref{vert-ln} of Algorithm \ref{propag}. By Lemmas \ref{stop-point-query-segment} and \ref{queryR}, we identify these wavelets in $O(\log n)$ time.

Secondly, the algorithm discovers the closest vertex to $\qp$, say $v$ (i.e., $v = \Gamma(\wv)$ in line \ref{first-ln}), as it is the first vertex hit by the wavelet. By Lemma \ref{queryR}, this vertex can be determined in $O( \log n )$ time. When $\wv$ hits $v$, the wavelet is split into three new wavelets. In Algorithm \ref{propag} and in Figure \ref{point-prop} (b), these wavelets are denoted by $\wv_1$, $\wv_2$ and $\wv_3$.
 We also create (at most) four new wavelets, corresponding to $v$ and its search regions (i.e., the wavelets in $\prs(\qp)$). Since all these operations are executed in $O(\log n)$ time, we can conclude that $PropagatePoint(\wv)$ takes $O(\log n)$ time.


\begin{algorithm}[H]
\caption{$PropagatePoint(\wv(\qp,t,\range))$} \label{propag}
\begin{description}
\item {\bf Input}: A point wavelet $\wv$.
\item {\bf Output}: Propagates $\wv$ and adds new wavelets to the queue if necessary. 
\end{description}
\begin{algorithmic}[1]
\STATE let $\range=(X_1, X_2) \times (Y_1, Y_2)$
\STATE let $\qp = ((x, y), t_0)$
\IF {$(x, y)$ is a vertex} 
\STATE \label{del-ver-ln} delete the vertex from $V$ and ``permanently'' label it with the value of $t$
\ENDIF
\STATE \label{as-ln} let $S$ be the set of accessible segments of $\qp$, on the boundary of $\range$
\par \COMMENT{since vertices are in general positions, we have $|S| \leq 4$}
\FOR{each $e \in S$, where $e$ is located on $E_i$}
\STATE let $\range^e$ be the segment search region of $e$
\STATE \label{seg-ln} add $\wv((e, T^d_i),T^d_i,\range^e)$ to the queue
\FOR{each end point $p$ of $e$}
\STATE \label{vert-ln} add the wavelets of $\prs((p, T^d_i))$ to the queue \par \COMMENT{the point wavelets originating from $(p, T^d_i)$}
\ENDFOR
\ENDFOR
\STATE w.l.o.g., assume $\wv$ is propagating north-east
\STATE \label{first-ln} let $v{=}(x', y') = \Gamma(\wv)$ \COMMENT{$v$ is the closest vertex to $\qp$ inside $\range$}
\STATE \label{min-arrive-ln} let $t' = t_0 + \Big({|x'-x|+|y'-y| \over \vr}\Big)$ \COMMENT{$t'$ is the minimum arrival time at $v$}
\STATE let $\qp' = (v, t')$
\STATE add the wavelets of $\prs(\qp')$ to the queue  \par \COMMENT{the point wavelets originating from $\qp'$}
\STATE let $\range_1=(X_1, x') \times (Y_1, Y_2)$ \COMMENT{bottom split}
\STATE \label{w1} add $\wv_1(\qp,t',\range_1)$ to the queue
\STATE let $\range_2=(X_1, X_2) \times (Y_1, y')$ \COMMENT{top split}
\STATE \label{w2} add $\wv_2(\qp,t',\range_2)$ to the queue
\STATE let $\range_3=(x', X_2) \times (y', Y_2)$ \COMMENT{middle split}
\STATE \label{w3} add $\wv_3(\qp',t',\range_3)$ to the queue
\end{algorithmic}
\end{algorithm}

We now describe an algorithm called $PropagateSegment(\wv)$ (see Algorithm \ref{propag2}). This algorithm takes a segment wavelet as input and propagates it inside its corresponding search region. There are four types of segment wavelets depending on their directions (i.e., $N$, $E$, $W$ and $S$). W.l.o.g., assume $\wv(\qs,t,\range)$ is propagating north. 
 The algorithm finds the north stop segment for $\qs$, denoted by $e'$. Note that, by Lemma \ref{stop-point-query-segment}, this can be done in $O(\log n)$ time.
 When $\wv$ hits $e'$, $\wv$ is split into smaller wavelets as follows. 
 A new wavelet $\wv'$ originating from $e'$, is added to the queue (in line \ref{new-wave-seg-ln} of Algorithm \ref{propag2}). Then, the algorithm adds smaller wavelets for the parts of $\wv$ which do not hit $e'$. As an example, see $\wv_1$ and $\wv_2$ in Figure \ref{seg-prop} (b). Additionally, if a vertex is hit by the segment wavelet, our algorithm adds (at most) four point wavelets originating from the vertex. By Lemma \ref{spqp}, this can be done in $O(\log n)$ time. Thus, $PropagateSegment(\wv)$ runs in $O(\log n)$ time.

\begin{algorithm}[H]
\caption{$PropagateSegment(\wv(\qs,t,\range))$} \label{propag2}

\begin{description}
\item {\bf Input}: A segment wavelet $\wv$.
\item {\bf Output}: Propagates $\wv$ and adds new wavelets to the queue if necessary.
\end{description}
\begin{algorithmic}[1]
\STATE w.l.o.g., assume $\wv$ is propagating north
\STATE let $\qs = (((x_1, y),(x_2,y)),t)$
\STATE let $e' = \spp(\qs, N)$  \COMMENT{$e'$ is the north stop segment for $\qs$ inside $\range$}
\STATE suppose $e' = ((x_1', y'),(x_2',y'))$, where $e'$ is locate on the edge $E_i$
\STATE \label{new-wave-seg-ln} let $\qs' = (e',T^d_i)$ 

\STATE let $\range'$ be the segment search region of $e'$
\STATE \label{add-seg-ln} add $\wv'(\qs',T^d_i,\range')$ to the queue

\IF {$(x_1', y')$ is a vertex} 
\STATE \label{4-w-ln-1}
add the wavelets of $\prs(((x_1', y'), t'))$ to the queue  \par \COMMENT{the point wavelets originating from $((x_1', y'), t')$}
\STATE let $\qs_1 = (((x_1, y'),(x_1',y')),t')$
\STATE let $\range_1$ be the segment search region of $\qs_1$ 
\STATE add $\wv_1(\qs_1,t',\range_1)$ to the queue
\ENDIF
\IF {$(x_2', y')$ is a vertex} 
\STATE \label{4-w-ln-2} add the wavelets of $\prs(((x_2', y'), t'))$ to the queue
\STATE let $\qs_2 = (((x_2', y'),(x_2,y')),t')$
\STATE let $\range_2$ be the segment search region of $\qs_2$ 
\STATE add $\wv_2(\qs_2,t',\range_2)$ to the queue
\ENDIF
\end{algorithmic}
\end{algorithm}

\begin{figure}[t]
\captionsetup[subfigure]{justification=centering}
\centering
 \begin{subfigure}{.45\linewidth}
 \centering
 \includegraphics[width=125pt]{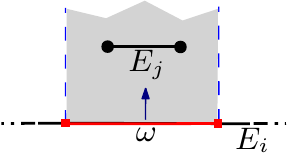}
  \caption{}
 \end{subfigure}
 \begin{subfigure}{.45\linewidth}\centering
 \includegraphics[width=125pt]{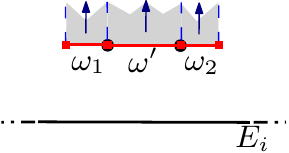}
  \caption{}
 \end{subfigure}
\caption{The process of a segment propagation is illustrated.} \label{seg-prop}
\end{figure}

\subsection{A Naive Algorithm} \label{nai-sec}

In this section, we present a ``naive'' algorithm (see Algorithm \ref{naive-alg}), which reports the minimum arrival time at the destination in $O(n^2 \log n)$ time. Although the algorithm is not efficient, it illustrates our global approach and serves as the basis
for our optimal algorithm describe later in Section \ref{eff-sec}.

In this algorithm, our approach is to find the minimum arrival time at every vertex in $V$ from the source. To achieve this, a set of wavelets is created and maintained in a priority queue, whose keys are their $t$ values (i.e., their corresponding departure times). The queue is initialized with four initial point wavelets originating from the start point $s$ in four directions $NE$, $NW$, $SE$ and $SW$. In each iteration, the algorithm proceeds by extracting a wavelet from the queue with lowest value of $t$. If the wavelet originates from the destination point, the minimum arrival time at the destination has been found. Otherwise, depending on whether $\wv$ is a point wavelet or a segment wavelet, $PropagatePoint(\wv)$ or $PropagateSegment(\wv)$ is executed, respectively. Recall that,
these algorithms propagate the given wavelet and add, or update, the wavelets in the queue, if necessary. In the following lemma, we prove that the naive algorithm correctly finds the minimum arrival at the destination.

\begin{algorithm}[H]
\caption{Naive Algorithm} \label{naive-alg}
\begin{description}
\item {\bf Input}: A set of transient edges $E$, a source point $s$ and a destination point $d$
\item {\bf Output}: The minimum arrival time at the destination point $d$
\end{description}
\begin{algorithmic}[1]

\STATE initialize an empty priority queue $Q$.
\STATE add the wavelets of $\prs((s, 0))$ to the queue  
\WHILE {$Q \not= \emptyset$}

  \STATE \label{extract-ln} extract a wavelet $\wv(q, t, \range)$ from the queue whose associated time $t$ is minimum

  \IF {$q=(p, t_0)$ is a point source}
    \IF {$p$ is the destination} 
      \STATE return $t_0$
    \ENDIF
    \STATE \label{pp-ln} $PropagatePoint(\wv)$ 
  \ENDIF
  \IF {$q=(e, t_0)$ is a segment source}
    \STATE \label{ps-ln} $PropagateSegment(\wv)$ 
  \ENDIF
\ENDWHILE
\end{algorithmic}
\end{algorithm}

\begin{lemma}\label{naiveproof}
The naive algorithm reports the minimum arrival time at the destination.
\end{lemma}

\begin{proof}
We prove that for every vertex $v \in V$, there is a wavelet $\wv((v, t_v),t_v,\range_v)$ in the priority queue (before time $t_v$), such that, when $\wv$ is extracted, $t_v$ is the minimum arrival time at $v$. In such case, we say that $v$ is ``assigned correctly''. At any given time, let $S$ be a set of vertices whose corresponding wavelets have been extracted from the queue. 
The proof follows by induction on $|S|$. It is easily seen that for $|S| \leq 2$ all vertices in $S$ are assigned correctly. 
For the inductive step, let $v$ be the last vertex added to $S$ and
assume that all vertices in $S \setminus \{v\}$ have
been assigned correctly. We now prove that our algorithm computes the minimum arrival time at $v$.

Let $\pi$ be a collision-free shortest path from $s$ to $v$. Let $u=(X_u, Y_u)$ and $v=(X_v, Y_v)$ be two consecutive vertices on $\pi$.
We denote by $\pi(u, v)$, the sub-path of $\pi$ from $u$ to $v$. Let $M$ be a rectangle where $u$ and $v$ are its two opposite corners. By the monotonicity property, observe that $\pi(u, v)$ is a monotone path inside $M$.
By the inductive hypothesis, there is a wavelet $\wv((u, t_u),t_u,\range_u)$ in the priority queue, where $t_u$ is the minimum arrival time at $u$. In the remaining of this proof, we assume $\wv$ is propagating north-east and $u$ is the bottom left corner of $r_u$. Other directions can be treated analogously.
We then consider two cases: \textbf{($i$)} $\pi(u, v)$ contains no wait points and \textbf{($ii$)} $\pi(u, v)$ contains wait points. In either case, we prove that our algorithm assigns $v$ correctly.

\begin{itemize}
\item [\textbf{($i$)}] Suppose $\pi(u, v)$ contains no wait points. 
 In line \ref{pp-ln} of Algorithm \ref{naive-alg}, $PropagatePoint(\wv)$ is executed to propagate $\wv$. Let $(X^N, Y^N)$ be the north stop point and $(X^E, Y^E)$ be the east stop point for $(u, t_u)$. 
 Again, we break down the problem into two scenarios: (a) when $v \in \range_u$ and (b) when $v \not\in \range_u$.
\begin{itemize}
 \item[(a)] 
 Assume $v$ is the closest vertex to $u$ inside $\range_u$, i.e., $v = \Gamma(\wv)$. Thus, in line \ref{min-arrive-ln} of Algorithm \ref{propag}, we have $t_v = t_u + \Big({|X_v- X_u|+|Y_v -Y_u| \over \vr}\Big)$ which is the minimum arrival time at $v$. 
Now, assume $v_1 = \Gamma(\wv)$, where $v_1 \not= v$. Then, in lines \ref{w1}, \ref{w2} and \ref{w3}, three new wavelets $\wv_1(\qp,t',\range_1)$, $\wv_2(\qp,t',\range_2)$ and $\wv_3(\qp',t',\range_3)$ are added to the queue, where $t'$ is the arrival time at $v_1$.
We derive the following observations for the three search regions $\range_1$, $\range_2$ and $\range_3$: (1) each is a sub-region of $\range_u$, (2) each contains fewer vertices than $\range_u$ (because $v_1$ is deleted in line \ref{del-ver-ln} of Algorithm \ref{propag}) and (3) one of them contains $v$. 
W.l.o.g., assume $\wv_1$ is the one whose search region contains $v$.
Later on, our algorithm extracts $\wv_1$ and thus, we have $v_2 = \Gamma(\wv_2)$. If $v_1 \not= v$, the above procedure is repeated again.
 Since there are at most $n$ vertices inside $r_u$ and one vertex is deleted in each iteration, after at most $O(n)$ iterations, there is a wavelet $\wv_h$ in the queue where $v = \Gamma(\wv_h)$. Let $\{u, v_1, v_2, ..., v_h, v\}$ be the sequence of the vertices in the above process. Since in each iteration, $v_i$ is located north-east of $v_{i-1}$, we have $t_v = t_u + \Big({|X_v- X_u|+|Y_v -Y_u| \over \vr}\Big)$, which is the minimum arrival time at $v$.

 \item[(b)] Suppose $v \not\in \range_u$. Thus, we have $Y^N < Y_v$ or $X^E < X_v$ (or both). W.l.o.g., assume $Y^N < Y_v$. Let $e$ be the north accessible segment for $\range_u$ (which is discovered in line \ref{as-ln} of Algorithm \ref{propag}), where $(X_1, Y_1)$ is its right end point. Note that, $e$ intersects the left edge of $M$. However, if $e$ also intersects the right edge of $M$, then, by Observation \ref{prepan-obs}, any path from $v$ to $u$ inside $M$, must have a wait point on $e$. 
 Since we assumed $\pi(u, v)$ contains no wait points, so, $(X_1, Y_1)$ must be located inside $M$ (i.e., south-west of $v$). 

Next, in line \ref{first-ln} of Algorithm \ref{propag}, a wavelet $\wv_1(((X_1, Y_1), t_1), t_1, \range_1)$ is added to the queue, where $t_1$ is the minimum arrival time at $(X_1, Y_1)$.
Now, if $v \not\in \range_1$, our algorithm repeats the above procedure until there is a wavelet $\wv_k$ in the queue, originating from $(X_k, Y_k)$, where $v \in \range_k$. Let $\{u, (X_1, Y_1), ..., (X_k, Y_k)\}$ be the sequence of the points, discovered in the above process between $u$ and $v$. Since in each iteration $(X_i, Y_i)$ is located north-east of $(X_{i-1}, Y_{i-1})$, we have $t_k = t_u + \Big({|X_k- X_u|+|Y_k -Y_u| \over \vr}\Big)$. Now that $v \in \range_k$, by Case (a), we have $t_v = t_u + \Big({|X_v- X_u|+|Y_v -Y_u| \over \vr}\Big)$, which is the minimum arrival time at $v$.

\end{itemize}

\item [\textbf{($ii$)}] Let $W=\{p_1, ..., p_h\}$ be the set of wait points on $\pi(u, v)$, sorted by their departure times. W.l.o.g., for each $p_i \in W$, let $E_i \in E$ be the ``horizontal'' edge where $p_i$ is located on. 
We derive the following observations:
(1) by Observation \ref{prepan-obs}, $\pi(p_1, v)$ is a vertical line segment which contains all the wait points, (2) $\pi(u, p_1)$ contains no wait points, and (3) the minimum arrival time at $v$ is $T^{d}_h + {\norm{ve_h} \over \vr}$, where $\norm{ve_h}$ is the distance between $v$ and $p_h$.

 Since $\pi(u, p_1)$ contains no wait points, similar to the ($i$), Algorithm \ref{propag} recursively adds a sequence of point wavelets to the priority queue, which eventually calculates the minimum arrival time at $p_1$. 
 Let $\wv_1((e_1, T^{d}_1),T^{d}_1,\range_1)$ be a segment wavelet where $p_1 \in e_1$. 
Thus, $PropagateSegment(\wv_1)$ is executed to propagate this wavelet.
 Observe that $p_2$ is located on the north stop segment for this wavelet. Thus, in line \ref{add-seg-ln} of Algorithm \ref{propag2}, $\wv_2((e_2, T^{d}_2), T^{d}_2,\range_2)$ is added to the queue. 
 Recursively, the algorithm repeats the above procedure until there is a wavelet $\wv_h((e_h, t_h),t_h,\range_h)$ in the queue, where $p_h \in e_h$. Since $v$ is a vertex located inside $\range_h$, in line \ref{4-w-ln-1} 
 of the Algorithm \ref{propag2}, $\wv((v, t_v),t_v,\range_v)$ is added to the queue where $t_v=t_h + {\norm{ve_h} \over \vr}$, which is the minimum arrival time at $v$. 
\end{itemize}
By \textbf{($i$)} and \textbf{($ii$)}, the naive algorithm assigns $v$ correctly, which completes the induction.
\end{proof}

\begin{lemma}\label{timepropag}
The naive algorithm runs in $O(n^2 \log n)$ time.
\end{lemma}

\begin{proof}
 We estimate the total number of calls to the functions $PropagatePoint(\wv)$ and $PropagateSegment(\wv)$. More precisely, we need to bound the total number of wavelets created in the process. Fujimora \cite{fuji1} proved that at any given time, the size of the wavefront (i.e., the number of wavelets in the priority queue) is $O(n)$. Thus, each edge may be hit by $O(n)$ wavelets and consequently generate $O(n)$ new wavelets. This means that the total number of wavelets is bounded by $O(n^2)$. Recall that, each propagation
 (i.e., Algorithms \ref{propag} or \ref{propag2}) 
 can be executed in optimal $O(\log n)$ time. Therefore, Algorithm \ref{naive-alg} runs in $O(n^2 \log n)$ time.
\end{proof}
\section{An Improved Algorithm} \label{eff-sec}

As we proved in Lemma \ref{timepropag}, there may be up to $O(n)$ wavelets originating from a single edge. In Section \ref{expand}, we will utilize a method called 
``expanding'', to reduce the total number of segment propagations. Note though that the queue may contain wavelets with overlapping search regions. Thus, each vertex may be hit by $O(n)$ wavelets (the maximum size of the queue at any given time). To prevent this from happening, in Section \ref{narrow}, we propose a procedure called ``Narrowing'', which shrinks the overlapping search regions, so that they do not sweep the same area.



\subsection{Wavelet expanding} \label{expand}
 
Let $\mathcal{W}_i$ be a maximal set of wavelets, originating from the body of $E_i$. Recall that, the naive algorithm propagates every wavelet in $\mathcal{W}_i$ individually. In Lemma \ref{timepropag}, we proved that the number of these wavelets in the priority queue will be quadratic in the worst case.
In this section, we propose an alternative approach in which, we replace all wavelets in $\mathcal{W}_i$ by a single ``expanded'' wavelet. This wavelet is a segment wavelet whose source is the body of the edge $E_i$. We will prove that 
this replacement of wavelets permits avoiding the
quadratic number of propagations. The crucial property that we are employing is the following:

\begin{observation}\label{wi}
The wavelets in $\mathcal{W}_i$ simulate the robot's motions when: (1) it arrives at the body of the edge $E_i$ in its existence time interval and (2) departs from the edge, at time $T_i^d$. In other words, the departure time of all wavelets in $\mathcal{W}_i$ is the disappearance time of $E_i$. 
\end{observation}


W.l.o.g, we assume $E_i$ is a horizontal edge. 
Let $\wv(\qp {=} (v, t), t, r)$ be a wavelet originating from a vertex $v$, propagating north-east. 
Suppose $\wv$ hits the interior of edge $E_i$, i.e., $\spp(\qp, N) \in E_i$.
Thus, $PropagatePoint(\wv)$ creates some wavelet(s) originating from the body of $E_i$ (for details see Algorithm \ref{propag}). The same process will be repeated for the newly generated wavelets, at a later time. 
Notice that the sequence at which these wavelets are created is sorted by their departure times. Thus, for each wavelet $\wv'$ originating from the body of an edge, we can find a sequence of wavelets, starting with the wavelet $\wv$ originating from a vertex, which led to the creation $\wv'$. We say $\qp$ is the \textit{root point source} of $\wv'$. Since $E_i$ is horizontal, we observe the following property of the root point sources.

\begin{property}\label{prop-bi}
Let $\wv(\qs, t, \range) \in \mathcal{W}_i$ be a segment wavelet where $\qs((X_1, Y),(X_2, Y))$. Then, there is an associated root point source $\qp {=} (v, t)$ for which: (a) $v{=}(X_v, Y_v)$ is a vertex in $V$ and (b) $X_1 \leq X_v \leq X_2$, i.e., $\qp$ is located south of $\qs$.
\end{property}

We store the root point sources of the wavelets of $\mathcal{W}_i$ in a binary search tree (BST) $B_i$. If $E_i$ is horizontal, $B_i$ is sorted by the $x$-coordinates of the point sources; and if $E_i$ is vertical, $B_i$ is sorted by their $y$-coordinates. 
Next, we use a method, called $Expand$, as follows: remove the wavelets of $\mathcal{W}_i$ from the queue and replace them with an expanded wavelet $\wv_i((E_i, T^d_i), T^d_i, r_i)$, where $r_i$ is the segment search region of $E_i$. The details of this procedure is presented in Algorithm \ref{expanding}. In this method, we update the binary trees as follows: 
 when $\wv_i$ hits its north stop segment on the edge $E_j$, we construct $B_j$ by applying the appropriate \textit{Split} and \textit{Merge} operations on $B_i$ and $B_j$. 
We define these two functions as follows.
\begin{itemize}[noitemsep,wide=0pt, leftmargin=\dimexpr\labelwidth + 2\labelsep\relax]
	\item[$\bullet$] $Split(T, x):BST \times {\rm I\!R} \rightarrow BST \times BST$. Given a BST $T$ and a key value $x$, $split$ divides $T$ into two BSTs $T_l$ and $T_r$, where $T_l$ consists of all point sources in $T$ with $x$-coordinates less than $x$; and $T_r$ includes the rest of the point sources. 
	\item[$\bullet$] $Merge(T_l, T_r):BST \times BST \rightarrow BST$. Let $T_l$ and $T_r$ be two BSTs, where there exist a value $x$ such that the point sources in $T_l$ have lower (or equal) $x$-coordinates than $x$ and the point sources in $T_r$ have greater (or equal) $x$-coordinates than $x$. Function $Merge$ creates a new BST which is the union of $T_l$ and $T_r$.
\end{itemize}

\begin{algorithm}[H]
\caption{$Expand(\wv(q, T^d_i, r))$} \label{expanding}
\begin{description}
\item {\bf Input}: A point wavelet $\wv \in \mathcal{W}_i$.
\item {\bf Output}: Replaces the wavelets of $\mathcal{W}_i$ by an expanded wavelet in the queue.
\end{description}
\begin{algorithmic}[1]
\STATE w.l.o.g., assume $E_i$ is horizontal
\STATE let $B_i$ be an empty binary search tree \par \COMMENT{ $B_i$ stores the root point sources of the wavelets in $\mathcal{W}_i$}
\FOR {each $\wv'(q', T^d_i, r') \in \mathcal{W}_i$}
 \STATE extract $\wv'$ from the queue 
  \IF {$\wv'$ is not an expanded wavelet} 
      \STATE insert the root point source of $\wv'$ into $B_i$
  \ENDIF
  \IF {$\wv'$ is an expanded wavelet of $E_j$}
   \STATE let $q' = (((X_l, Y),(X_r, Y)), T^d_j)$
      \STATE $T_l, T_r = Split(B_j, X_r)$
      \STATE $T'_l, T'_r = Split(T_l, X_l)$ 
      \STATE $B_i = Merge(B_i, T'_r)$
      \STATE $B_j = Merge(T_r, T'_l)$
  \ENDIF
\ENDFOR
\STATE let $r_i$ be a segment search region of $E_i$
\STATE add $\wv_i((E_i, T^d_i), T^d_i, r_i)$ to the queue \par \COMMENT{ the expanded wavelet is added instead of the wavelets originating from $E_i$}
\end{algorithmic}
\end{algorithm}

In Algorithm \ref{expanding} (i.e., $Expand(\wv)$), the input is a wavelet $\wv$ originating from the body of an edge $E_i$. The algorithm proceeds by initializing an empty BST $B_i$. Next, for any wavelet $\wv'$ originating from the body of $E_i$, it first removes $\wv'$ from the queue. Then, it considers two cases: (1) if $\wv'$ is a non-expanded wavelet, it inserts the root point source of $\wv'$ into $B_i$ and (2) if $\wv'$ is an expanded wavelet of edge $E_j=((X_l, Y), (X_r, Y))$, the algorithm first splits $B_j$ into two sub-trees $T_l, T_r = Split(B_j, X_r)$. Then, it splits $T_l$ again, such that $T'_l, T'_r = Split(T_l, X_l)$. At this point, $T'_r$ represents the point sources in $B_j$ with $x$-coordinates between $X_l$ and $X_r$. 
The algorithm merges $T'_r$ with $B_i$ using $B_i = Merge(T'_r, B_i)$.
The rest of the point sources in $T_r$ and $T'_l$ are maintained in $B_j$ using $B_j = Merge(T_r, T'_l)$.

Since updating the binary search trees using the basic operations of $merge$ and $split$, can be done in $O(\log n)$ time, each iteration in the main loop of the Algorithm \ref{expanding} runs in $O(\log n)$ time. 


Now, we modify the naive algorithm described in Section \ref{nai-sec} so that it uses the ``expanding'' method. Let $\wv$ be a wavelet extracted from the priority queue (see Line \ref{extract-ln} of the Algorithm \ref{naive-alg}). If $\wv$ is originating from the body of $E_i$,
we execute the Algorithm \ref{expanding} to replace the wavelets of $\mathcal{W}_i$ with an expanded wavelet. More precisely, if $\wv$ is originating from the body of $E_i$, we execute $Expand(\wv)$ right after line \ref{extract-ln} of the Algorithm \ref{naive-alg}. We call this new algorithm \textit{Expanding algorithm}. 

W.l.o.g., assume the expanded wavelet $\wv_i$ is propagating north and $v{=}(X_v, Y_v)$ is the first vertex that it hits.
Although $PropagateSegment(\wv_i)$ identifies $v$, it may not report the arrival time at $v$. This is due to the fact that the  wavelets on the body of $E_i$ have been replaced by a single wavelet $\wv_i$. Thus, we need an alternative approach to calculate the minimum arrival time at $v$.
Suppose $\wv_{min}$ is the first wavelet in $\mathcal{W}_i$ that hits $v$ (see Figure \ref{w-expanding}). In the remaining of this section, we show how to determine the minimum arrival time at $v$, without explicitly calculating the wavelet $\wv_{min}$. Recall that, the minimum $L_1$ distance from $u$ to $v$ is denoted by $\norm{uv}$, and the minimum $L_1$ distance between $E_i$ and $v$ is denoted by $\norm{E_iv}$.

\begin{figure}[t]
\centering
\includegraphics[width=4in]{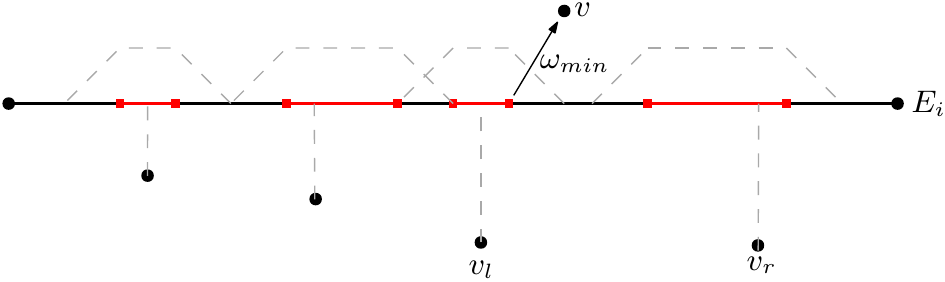}
\caption{The segment sources and point sources of an edge $E_i$ are illustrated by red segments and red squares, respectively. Among the wavelets originating from these sources (i.e., the wavelets in $\mathcal{W}_i$), $\wv_{min}$ is the first wavelet that hits vertex $v$.}\label{w-expanding}
\end{figure}

\begin{lemma}\label{seg-tree} 
Let $\qp(u, t) \in B_i$. If the robot departs from $u$ at time $t$, the minimum arrival time at $v$ is $T_v(u) = max\big(T_i^d , t + {{\norm{uv} - \norm{E_iv}} \over \vr}\big) + {\norm{E_iv} \over \vr}$.
\end{lemma}

\begin{proof} Let $\pi(u, v)$ be a time-minimal path from $u$ to $v=(X_v, Y_v)$, where the robot leaves $u$ at time $t$. Now, consider two cases:

\noindent \textbf{Case 1.} If $\pi(u, v)$ contains no wait points, then it is easily seen that
the robot moves with maximum speed (without waiting) along the path $\pi(u, v)$. Recalling that $\pi(u, v)$ is a $L_1$ monotone path (by the monotonicity property), the arrival time at $v$ is $T_v(u) = t + {{\norm{uv}} \over \vr}$.

\noindent  \textbf{Case 2.} Let $\pi(u, v)$ has a wait point $p{=}(X_v, Y_p)$ on some (horizontal) edge $E_j$. Then, by Observation \ref{prepan-obs}, the robot leaves $p$ at time $T^d_j$ and moves north (i.e., perpendicularly to $E_j$), towards $v$.
 Recall that, our underlying assumption is that $v$ is hit by a wavelet $\wv_{min} \in \mathcal{W}_i$. Thus, by Observation \ref{wi}, a time minimal path from $u$ to $v$ arrives at $E_i$ in its existence time interval. Hence, the robot arrives at some point $p'{=}(X_v, Y_{p'}) \in E_i$ before $T_i^d$. Note that the $X$-coordinate of $p$ is equal to the $X$-coordinate of $v$. So, since the path $\pi(p', v)$ is oriented perpendicular to $E_i$, the robot has to wait at $p'$ until time $T_i^d$. There is a crucial observation to make here:  $\pi(u, v)$ either has no wait points, or has a wait point on $E_i$.
  Therefore, when $\pi(u, v)$ has a wait point, the minimum arrival time at $v$ is $T_v(u) = T_i^d + {{\norm{E_iv}} \over \vr}$.

When there are no wait points before arriving at $p'$, observe that the minimum arrival time at $p'$ is $t' = t + {{\norm{uv} - \norm{E_iv}} \over \vr}$. In the following, we show how to determine if $\pi(u, v)$ has a wait point on edge $E_i$, by solely comparing the values of $t'$ and $T_i^d$. If $t' > T_i^d$, it means the robot cannot arrive at $p'$ before the disappearance of the edge. Since $p'$ is the only designated wait point on $E_i$, we conclude that $\pi(u, v)$ contains no wait points. So, by Case 1, $T_v(u) = t + {{\norm{uv}} \over \vr}$. Conversely, if $t' < T_i^d$, this implies that even if there are no wait points before $p'$, the robot has to wait at $p'$ until time $T_i^d$. This is due to the fact that the underlying assumption is that $v$ is hit by a wavelet $\wv_{min} \in \mathcal{W}_i$.
 Hence, by Case 2, the arrival time is $T_i^d + {{\norm{E_iv}} \over \vr}$. Therefore, the minimum arrival time at $v$ is $T_v(u) = max\big(T_i^d , t'\big) + {\norm{E_iv} \over \vr}$, in which $T_v(u) = t + {{\norm{uv}} \over \vr}$ if $t' > T_i^d$ and $T_v(u) = T_i^d + {{\norm{E_iv}} \over \vr}$, otherwise. 
\end{proof}


 Let $\qp_l=(v_l, t_l)$ be the point source in $B_i$, where $v_l$ is a vertex with the largest $X$-coordinate smaller than $X_v$. For an example, see Figure \ref{w-expanding}. 
 Analogously, let $\qp_r=(v_r, t_r)$, where $v_r$ has the smallest $X$-coordinate greater than $X_v$.
Using these notations, we state the following lemma.

\begin{lemma}\label{expand-proof} 
Among the point sources in $B_i$, either point source $\qp_r$ or $\qp_l$ (defined above) has the minimum arrival time at $v$.
\end{lemma}
\begin{proof} The proof is based on a simple observation: when the robot has arrived at $E_i$ (before time $T^d_i$), it moves on the edge, in a direction that would get closer to $v$. This will minimize the arrival time at $v$ after the disappearance of the obstacle. Thus, at time $T_i^d$, a wavelet in $\wv_{min} \in \mathcal{W}_i$ which has the closest distance to $v$ is the candidate for a time-minimal path. 

 By Observation \ref{wi}, the departure times of the wavelets of $\mathcal{W}_i$ are equal.
So,
 we can group every 
pair of segment wavelets 
in  $\mathcal{W}_i$ which have intersecting 
segment sources into one (longer) segment wavelet.
 So, w.l.o.g., we assume the segment wavelets in $\mathcal{W}_i$ have disjoint segment sources (see Figure \ref{w-expanding}). Hence, we can sort the wavelets in $\mathcal{W}_i$ based on the $x$-coordinates of their sources. 

 By Property \ref{prop-bi}, each segment wavelet $\wv(\qs, t, \range) \in \mathcal{W}_i$ has associated with it a root point source which is located south of $\qs$. Thus, it is easily seen that the order of the wavelets in $\mathcal{W}_i$ on edge $E_i$, is the same order as of their corresponding root point sources in $B_i$ (based on their $x$-coordinates). Finally, we observe that either $\qp_l$ or $\qp_r$ is the root source point for $\wv_{min}$ and therefore, either $T_v(v_l)$ or $T_v(v_r)$ has the minimum value.
\end{proof}

Finally, by Lemmas \ref{seg-tree} and \ref{expand-proof}, when a vertex $v$ is discovered by an expanded wavelet in Algorithm \ref{propag2}, we can find the minimum arrival time at $v$ using $t_v = min(T_v(v_r), T_v(v_l))$. Since the calculation of $t_v$ is solely based on $E_i$, $u$ and $v$, it is not required to calculate the wavelets in $\mathcal{W}_i$.
 Therefore, the following corollary is obtained.

\begin{corollary}\label{col-2} 
The Expanding algorithm calculates the minimum arrival time at the destination.
\end{corollary}

\subsection{Wavelet narrowing} \label{narrow}
In the previous section, we described a technique to reduce the number of wavelets originating from the edges by a factor of $n$. Here, we need to address another challenge: reducing the total number of vertex-originated wavelets. As mentioned before, the search regions may overlap and hence, a vertex may create $O(n)$ point wavelets. 
To prevent this from happening, in this section, we propose a procedure called ``Narrowing'' the wavelets.

Let $\wv_1(\qp_1, t_1, \range_1)$ and $\wv_2(\qp_2, t_2, \range_2)$ be a pair of point wavelets whose search regions intersect (i.e., $\range_1 \cap \range_2 \not= \emptyset$). W.l.o.g., assume that $\wv_1$ and $\wv_2$ are propagating toward the north-east. We denote by $p$, the bottom left corner of $\range_1 \cap \range_2$ (see Figure \ref{narrow-fig}). Recall that $T(\qp, p)$ is the minimum arrival time at $p$ from $\qp$. Let $T_1 = t_1 + T(\qp_1, p)$ and $T_2 = t_2 + T(\qp_2, p)$ be the minimum arrival times at $p$ from $\qp_1$ and $\qp_2$, respectively. 
The wavelet which arrives at $p$ first is called the \textit{dominant wavelet}. W.l.o.g., assume $T_1 < T_2$ and thus, $\wv_1$ is dominant.
Now, let $v \in \range_1 \cap \range_2$ be a vertex located inside the intersection of the two search regions. Note that, $\wv_1$ hits $v$ at time $T(\qp_1, v) = T(\qp_1, p) + {\norm{pv} \over \vr}$; and $\wv_2$ hits $v$ at time $T(\qp_2, v) = T(\qp_2, p) + {\norm{pv} \over \vr}$. Since $T(\qp_1, p) < T(\qp_2, p)$, we obtain the following:

\begin{lemma}\label{xy}
Let $\wv_1(\qp_1, t_1, \range_1)$ and $\wv_2(\qp_2, t_2, \range_2)$ be a pair of point wavelets, where $\wv_1$ is the dominant wavelet. For any vertex $v \in \range_1 \cap \range_2$, we have $T(\qp_1, v) < T(\qp_2, v)$.
\end{lemma}

By the above lemma, 
for any vertex $v \in \range_1 \cap \range_2$, the point source $\qp_2$ cannot be on a shortest path from $s$ to $v$. As a result, it is counter-productive to propagate $\wv_2$ inside $\range_1 \cap \range_2$. Intuitively, we can avoid this by replacing $\wv_2$ with new wavelets that are designated to sweep only inside $\range_1 \setminus \range_2$ (i.e., the areas inside $\range_1$ and outside $\range_2$). 
This procedure is called \textit{Narrowing}. Since the underlying search regions are axis-parallel rectangles, we can narrow a wavelet by replacing its search region by at most four smaller rectangles. As an illustration, in Figure \ref{narrow-fig}, $\wv_2$ is replaced with two new wavelets $\wv'_2$ and $\wv'_3$ with smaller search regions. The details of this procedure are presented in
Algorithm \ref{narrowing}. 
\begin{figure}[t]
\captionsetup[subfigure]{justification=centering}
\centering
 \begin{subfigure}{.4\linewidth}
 \centering
 \includegraphics[width=105pt]{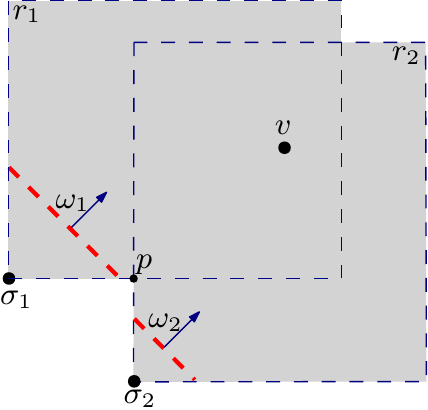}
  \caption{}
 \end{subfigure}
 \begin{subfigure}{.4\linewidth}\centering
 \includegraphics[width=105pt]{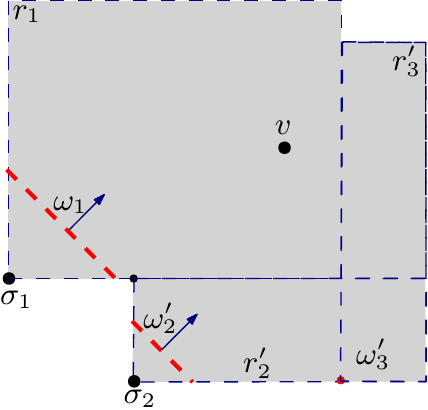}
  \caption{} 
 \end{subfigure}
\caption{An example of narrowing a wavelet; (a) the minimum arrival time at $p$ is identified from two point sources $\qp_1$ and $\qp_2$. The path from $\qp_1$ is faster, so $\wv_1$ is the dominant wavelet; (b) $\wv_2$ is replaced with two wavelets $\wv'_2$ and $\wv'_3$ with smaller search regions.} \label{narrow-fig}
\end{figure} 

\begin{algorithm}[H]
\caption{$Narrow(\wv_1(\qp_1, t_1, \range_1), \wv_2(\qp_2, t_2, \range_2))$} \label{narrowing}
\begin{description}
      \setlength\itemsep{-0.2em}
\item {\bf Input}: Two point wavelets $\wv_1$ and $\wv_2$.
\item {\bf Output}: Narrows $\wv$ and adds new wavelets to the queue if necessary.
\end{description}
\begin{algorithmic}[1]
\IF {$\range_1 \cap \range_2 = \emptyset$} 
 \STATE return
\ENDIF
\STATE w.l.o.g, assume $\wv_1$ and $\wv_2$ are propagating north-east and $\wv_1$ is the dominant wavelet
\STATE let $\range_1=(X_1^w, X_1^e) \times (Y_1^s, Y_1^n)$ and $\range_2=(X_2^w, X_2^e) \times (Y_2^s, Y_2^n)$
\STATE remove $\wv_2$ from the queue
\IF {$X_2^w < X_1^w$} 
 \STATE $\range_1'{=}(X_2^w, X_1^w) \times (Y_2^s, Y_2^n)$
 \STATE add $\wv'_1(\qp_2, t_2, \range_1')$ to the queue
\ENDIF
\IF {$Y_2^s < Y_1^s$} 
 \STATE $\range_2'{=}(X_2^w, X_2^e) \times (Y_2^s, Y_1^s)$
 \STATE add $\wv'_2(\qp_2, t_2, \range_2')$ to the queue
\ENDIF
\IF {$X_2^e > X_1^e$} 
 \STATE $\range_3'{=}(X_1^e, X_2^e) \times (Y_2^s, Y_2^n)$
 \STATE add $\wv'_3(\qp_2, t_2, \range_3')$ to the queue
\ENDIF
\IF {$Y_2^n > Y_1^n$} 
 \STATE $\range_4'{=}(X_2^w, X_2^e) \times (Y_1^n, Y_2^n)$
 \STATE add $\wv'_4(\qp_2, t_2, \range_4')$ to the queue
\ENDIF
\end{algorithmic}
\end{algorithm}

Our approach for reducing the number of point wavelets is based on identifying the dominant wavelets in each iteration. 
One greedy approach is to compare all pairs of the wavelets and narrow the non-dominant ones. However, this may result in quadratic number of narrowings. Our alternative approach is to execute the narrow procedure for every pair of wavelets originating from two vertices $v_1$ and $v_2$, when: (1) there exists two wavelet $\wv_1$ and $\wv_2$ originating from $v_1$ and $v_2$, respectively; and (2) either $\wv_1$ hits $v_2$, or $\wv_1$ and $\wv_2$ hit the same vertex $v_3$. This modification of the Expanding algorithm results in a new algorithm which we call the \textit{Narrowing algorithm}. 
The following corollary is a direct consequence of Lemma \ref{xy}. 

\begin{corollary}\label{col-1} 
The Narrowing algorithm calculates the minimum arrival time at the destination.
\end{corollary}

In order to prove that the Narrowing algorithm runs in $O(n \log n)$ time, we first establish a linear bound on the number of point wavelets created in the process.

\begin{lemma}\label{narrowed-proof} 
The total number of point wavelets created by the Narrowing algorithm is $O(n)$.
\end{lemma}
\begin{proofsketch} Intuitively, we prove the following: Let $\wv_1$ and $\wv_2$ be two north-east propagating wavelets, originating from $u_1$ and $u_2$, respectively.
 Assume $\wv_1$ and $\wv_2$ both hit a vertex $v$. If $u_2$ lies within a $\norm{u_1v}$ distance from $u_1$ (in any direction), then either $\wv_1$ or $\wv_2$ would be narrowed before hitting $v$, making it impossible to have encountered $v$. So, there are at most two wavelets hitting $v$ from south-west. Therefore each vertex is hit by a constant number of wavelets in total. 
 \end{proofsketch}
 
\begin{proof} 
 In the naive algorithm, each vertex $v \in V$ has associated with it a list of vertices $Z^{NE}{=}\{u_1, $ $...,u_k\}$ that are south-east of $v$. For each vertex $u \in $ there exist a wavelet $\wv_u$ originating from $u$, which is: (1) propagating north-east and (2) hitting vertex $v$. 
 We assume the wavelets in $Z^{NE}$ are sorted in order of increasing $y$-coordinates of their point sources (observe that the order of increasing $y$-coordinates is the same as the order of increasing $x$-coordinates, since the point sources will form a staircase path going north-east). Similar definitions apply to $Z^{NE}$, $Z^{SE}$, and $Z^{SW}$. In the following, we prove that $Z^{NE}$ contains at most two wavelets in the Narrowing algorithm.

For simplicity, we assume $v{=}(0, 0)$. Let $Z_1^{NE}=\{u_1, ..., u_m\}$ be a subset of $Z_1^{NE}$ where for each vertex $u_i{=}(X_{i}, Y_i) \in Z_1^{NE}$ we have $Y_i < X_i$. Let $u_i, u_{j} \in Z_1^{NE}$ be a pair where $1 \leq i < j \leq m$. So, we have $\norm{u_iv} = |X_i| + |Y_i|$ and $\norm{u_iu_j} = |X_i - X_{j}| + |Y_i - Y_{j}|$. It is easily seen that $\norm{u_iu_j} < \norm{u_iv}$. 
Let $\wv_i$ and $\wv_{j}$ be two wavelets originating from $u_i$ and $u_{j}$ which are propagating north-east. 
Let $\wv'_i$ be a point wavelets originating from $u_i$ which is propagating south-east. Since $\norm{u_iu_j} < \norm{u_iv}$, it is easily seen that $\wv'_i$ hits $u_{j}$ before $\wv_i$ hits $v$. Thus, for each pair of wavelets $\wv_i$ and $\wv_{j}$, the non-dominant wavelet is narrowed before a point wavelet in $Z_1^{NE}$ hitting the vertex $v$. Therefore, there is at most one wavelet originating from a vertex in $Z_1^{NE}$ hitting $v$ (the one that is dominant over other wavelets in $Z_1^{NE}$). 

Using a similar argument, there is at most one (dominant) wavelet originating from a vertex in $Z_2^{NE}=\{u_{m+1},..., u_k\}$ which will hit $v$. Thus, there are at most two point wavelets in $Z^{NE}$ hitting the vertex $v$.
\end{proof}
 %

By the above lemma, there are $O(n)$ point wavelets in the queue. Furthermore, in Section \ref{expand}, we proved that the total number of segment wavelets is also $O(n)$. Thus, the running time of the Narrowing algorithm is $O(n \log n)$. Finally, by recording the sequence of the propagations during the process, we can actually construct the time-minimal path among the transient obstacles. Note that, similarly to the optimality argument for the existing $\Theta(n \log n)$ time algorithm \cite{rec} for the non-transient obstacles (which is a special case of our problem), our algorithm is also optimal.

\begin{theorem}\label{final-the} 
A time-minimal rectilinear path among transient rectilinear segments can be found in $\Theta(n \log n)$ time.
\end{theorem}


As the algorithm proceeds, by recording the trace of the endpoints of the wavelets, we can build a subdivision of the plane. Since the size of this subdivision is proportional to $n$, by \cite{subdivision}, we can construct a data structure to answer point location queries in $O(\log n)$ time. Thus, we can build the shortest path map with respect to a fixed source point in $O(n \log n)$ time. Now, for a given query point $q$, the minimum arrival time at $q$ from $s$, can be reported in $O(\log n)$ time.

\begin{theorem}\label{final-query} 
Given a set of $n$ transient edges $E$, a fixed source point $s$ and a query point $q$,  $E$ can be preprocessed in $O(n \log n)$ time, so that the minimum arrival time at $q$ can be reported in $O(\log n)$ query time.
\end{theorem}

\section{Extensions}

We present a few natural generalizations of our algorithm which are worth looking at.

\textbf{Rectilinear Transient Obstacles.} In case the obstacles are rectilinear simple polygons, we can consider the edges of these obstacles as transient segment obstacles. Note that in this case, we need to remove the wavelets that sweep the interior of the obstacles, during their existence interval. These wavelets can be identified in constant time by comparing the arrival/departure time of each wavelet to the existence interval of its corresponding obstacle.

\textbf{Time-Minimal Paths with Fixed Orientations.} The $L_1$ metric is the special case where the edges of the path are in four fixed orientations: $0, \pi/2,
\pi$, and $3\pi/2$. Given a set of $k$ fixed orientations, a $k$-oriented path is a polygonal path with each edge parallel to one of the given orientations (see \cite{fixed-ori}). In such case, the wavelets are also in $k$ inclinations. When $k$ is constant, a generalization of the propagation algorithms (see Section \ref{tcp}) is sufficient to handle these wavelets, without any significant changes.
Since the Euclidean distance can be approximated within an
accuracy of $O(1/k^2)$ \cite{fixed-ori} with $k$-oriented distance, the shortest paths with fixed orientations can be used as a basis to approximate a Euclidean shortest path. However, our algorithm does not immediately guarantee an approximation for a  time-minimal path. This is due to the fact a small delay on the path can create long waiting points in the future. We leave this interesting extension as a future work.

\textbf{Vertical Ray Shooting.} A problem closely related to finding the stop points and stop segments is dynamic vertical ray shooting. In that problem, we are given a set of dynamic planar segments on the plane, and we are interested in finding the first edge lying above a given query point $p$. We can imagine shooting a vertical ray to the north from a source point, and we want to return the first edge which it intersects. The dynamic vertical ray shooting has been solved 
\cite{dyn-ray} in $O(\log n)$ query time
after an $O(n \log n)$ preprocessing time. However, in our problem, the robot is limited by its maximum speed. Thus, this problem is a generalization of vertical ray shooting to a case where there is a maximum speed for the shooting ray. 

\bibliographystyle{plain}
\bibliography{thesis}

\begin{thebibliography}{10}

\bibitem{Agarwal-stab}
P.~K. Agarwal, L.~Arge, and K.~Yi.
\newblock An optimal dynamic interval stabbing-max data structure?
\newblock SODA '05, pages 803--812, 2005.

\bibitem{removeable}
P.~K. Agarwal, N.~Kumar, S.~Sintos, and S.~Suri.
\newblock {Computing Shortest Paths in the Plane with Removable Obstacles}.
\newblock In {\em SWAT 2018}, pages 5:1--5:15, 2018.

\bibitem{Chazelle-MST}
Bernard Chazelle.
\newblock Functional approach to data structures and its use in
  multidimensional searching.
\newblock {\em SIAM J. Comput.}, 17(3):427--462, June 1988.

\bibitem{origi-l1}
P.~J. de~Rezende, D.~T. Lee, and Y.~F. Wu.
\newblock Rectilinear shortest paths with rectangular barriers.
\newblock SCG '85, pages 204--213, New York, NY, USA, 1985. ACM.

\bibitem{fuji0}
K.~Fujimura.
\newblock On motion planning amidst transient obstacles.
\newblock In {\em ICRA 1992}, pages 1488--1493 vol.2, May 1992.

\bibitem{fuji-pixel}
K.~Fujimura.
\newblock Motion planning using transient pixel representations.
\newblock In {\em ICRA 1993}, pages 34--39 vol.2, May 1993.

\bibitem{fuji1}
K.~Fujimura.
\newblock Motion planning amid transient obstacles.
\newblock {\em The International Journal of Robotics Research}, 13(5):395--407,
  1994.

\bibitem{dyn-ray}
Y.~Giora and H.~Kaplan.
\newblock Optimal dynamic vertical ray shooting in rectilinear planar
  subdivisions.
\newblock {\em ACM Trans. Algorithms}, 5(3):28:1--28:51, July 2009.

\bibitem{violation}
J.~Hershberger, N.~Kumar, and S.~Suri.
\newblock {Shortest Paths in the Plane with Obstacle Violations}.
\newblock In {\em ESA 2017}, pages 49:1--49:14, 2017.

\bibitem{suri-cd}
J.~Hershberger and S.~Suri.
\newblock An optimal algorithm for euclidean shortest paths in the plane.
\newblock {\em SIAM Journal on Computing}, 28(6):2215--2256, 1999.

\bibitem{subdivision}
D.~Kirkpatrick.
\newblock Optimal search in planar subdivisions.
\newblock {\em SIAM Journal on Computing}, 12(1):28--35, 1983.

\bibitem{lavalle_2006}
S.~M. LaValle.
\newblock {\em Planning Algorithms}, pages 495--558.
\newblock Cambridge University Press, 2006.

\bibitem{LaValle}
S.~M. LaValle and R.~Sharma.
\newblock Robot motion planning in a changing, partially predictable
  environment.
\newblock In {\em ISIC 1994}, pages 261--266, Aug 1994.

\bibitem{Mi2}
D.T. Lee, C.D. Yang, and C.K. Wong.
\newblock Rectilinear paths among rectilinear obstacles.
\newblock {\em Discrete Applied Mathematics}, 70(3):185 -- 215, 1996.

\bibitem{Mi3}
J.~S. Mitchell.
\newblock L1 shortest paths among polygonal obstacles in the plane.
\newblock {\em Algorithmica}, 8(1-6):55--88, December 1992.

\bibitem{mitchell-cd}
J.~S.~B. Mitchell.
\newblock Shortest paths among obstacles in the plane.
\newblock SCG '93, pages 308--317, New York, NY, USA, 1993. ACM.

\bibitem{rectintersection}
V.K. Vaishnavi and D.~Wood.
\newblock Rectilinear line segment intersection, layered segment trees, and
  dynamization.
\newblock {\em Journal of Algorithms}, 3(2):160 -- 176, 1982.

\bibitem{fixed-ori}
P.~Widmayer, Y.~Wu, and C.~Wong.
\newblock On some distance problems in fixed orientations.
\newblock {\em SIAM Journal on Computing}, 16(4):728--746, 1987.

\bibitem{rec}
C.~Yang, D.~Lee, and C.~Wong.
\newblock Rectilinear path problems among rectilinear obstacles revisited.
\newblock {\em SIAM Journal on Computing}, 24(3):457--472, 1995.

\end{thebibliography}

\end{document}